\documentclass{llncs}
\usepackage{microtype}
\usepackage{amsmath}
\usepackage{amsfonts}
\usepackage{physics}
\usepackage{algorithm,multicol}
\usepackage{algpseudocode}

\newcommand {\secu} {$\lambda$ }
\newcommand {\poly} {\textsc{poly}($\lambda$)}
\newcommand {\expo} {\textsc{exp}($\lambda$)}
\newcommand {\negl} {\textsc{negl}($\lambda$)}
\newcommand {\nonegl} {\textsc{non-negl}($\lambda$)}
\newcommand {\udist} {\mathcal{U} }
\newcommand {\setup} {\texttt{Setup} }
\newcommand {\kgen} {\texttt{KeyGen} }
\newcommand {\enc} {\texttt{Enc} }
\newcommand {\dec} {\texttt{Dec} }
\newcommand {\qenc} {$\mathcal{QE}$ }
\newcommand {\qdec} {$\mathcal{QD}$ }
\newcommand*\colvec[3][]{ \begin{bmatrix}\ifx\relax#1\relax\else#1\\\fi#2\\#3\end{bmatrix} }

\begin{document}

\title{A Quantum-Classical Scheme towards \\ Quantum Functional Encryption}
\author{Aditya Ahuja}
\institute{Department of Computer Science and Engineering, \\ 
  Indian Institute of Technology Delhi \\
  \texttt{aditya.ahuja@cse.iitd.ac.in}}

\maketitle

\begin{abstract}

Quantum encryption is a well studied problem for both classical and quantum information. However, little is known about quantum encryption schemes which enable the user, under different keys, to learn different functions of the plaintext, given the ciphertext. In this paper, we give a novel one-bit secret-key quantum encryption scheme, a classical extension of which allows different key holders to learn different length subsequences of the plaintext from the ciphertext. We prove our quantum-classical scheme secure under the notions of quantum semantic security, quantum entropic indistinguishability, and recent security definitions from the field of functional encryption.
\end{abstract}

\section{Introduction}

In a pioneering work, Boneh, Sahai and Waters formalized the notion of functional encryption in 2010 \cite{fe}. This generalization of an encryption scheme enabled users possessing different keys to learn different functions over the plaintext from the ciphertext. There have been many (classical) schemes proposed to realize different variants of functional encryption since the inception of the notion. However, despite it's power, there have been no equivalent formalizations of functional encryption for quantum information. \\
Over the last few years, formal definitions of quantum entropic security \cite{entsec} and quantum computational security \cite{compsec} have been introduced and accepted. Moreover, for classical functional encryption schemes, refined security definitions \cite{fesec} have been recently presented, which extend the notion of security given in the original paper \cite{fe}. \\
In this work, we first present a one-qubit secret-key quantum encryption scheme for classical information. The scheme is proven secure under quantum semantic security (Definition 8, \cite{compsec}) and quantum entropic indistinguishability (Definition 3, \cite{entsec}). The classical, functional extension of this scheme is then proven to be full-message private, full-function private (Definitions 2.4,3.2 \cite{fesec}) and weakly simulation-secure (Definition 5, \cite{fe}). Intuitively, the security of the quantum encryption scheme and it's extension are based upon distinguishing computationally between two different but uniformly distributed bits, which is a hard problem. Also the functional extension allows the user to learn different length subsequences of the message, with a different subsequence per instantiation of the scheme. \\
This paper is organized as follows. First, we present the syntax - the notations and definitions used to present our scheme, in the preliminaries section. The third section of the paper gives the construction and correctness arguments for our scheme. The proofs of security are given in the following section. We then give the operational aspects of the scheme in the discussion section. To catalogue other works in quantum encryption, we give next a section on related work. In the final section, we present our conclusions and future extensions possible for this work.
\section{Preliminaries}
\label{sec:prelim}

In this section, we define the syntax and the definitions used for presenting our cryptosystem and establishing it's correctness and security.

\subsection{Notation}
\label{subsec:notation}

Let \secu be the security parameter. Let \poly, \expo, \negl, \nonegl \text{ } denote the set of all polynomial, exponential, negligible, and non-negligible functions on \secu respectively. We will sometimes abuse notation and place these function classes in place of functions that belong to these classes. \\
\noindent Let $\{ \ket{0},\ket{1} \}$ be the computational basis for our QPT algorithms. For both our QPT and PPT algorithms, we will use \texttt{A$^O$} to denote that \texttt{A} has oracle access to \texttt{$O$}. We will use $\udist(V)$ to denote uniform distribution on universe $V$.

\subsection{Correctness Definitions}
\label{subsec:cordef}

We begin by defining the functionality to be realized and the structure of our scheme. These definitions are inspired from \cite{fe}, but modified appropriately.

\begin{definition}[Functionality]
A functionality $F$, given $a \in A$, defined over $(K,M)$ is a function $F_a: K \times M \rightarrow \{0, 1\}^*$ describable as a deterministic Turing Machine. We call $A$ the functionality-index space, $M$ the message space, and $K \cup \{ \aleph \}$ the key space. \\
We introduce $\aleph$ to allow $\forall a \in A, \forall m \in M, F_a(\aleph,m) = |m|$.
\end{definition}

\begin{definition}[Hybrid FE Scheme]
A hybrid Functional Encryption (hFE) scheme $(\Pi,\Xi)$ for a functionality $F$, given functionality-index $a$, defined over $(K,M)$, is a tuple of PPT algorithms $\Pi = (\setup,\kgen,\enc,\dec)$ and a secret-key quantum encryption scheme $\Xi = (\mathcal{QE},\mathcal{QD})$ where \enc and \dec have oracle access to \qenc and \qdec respectively. \\
\newpage
\noindent This scheme must satisfy the following correctness condition given $A, \forall k \in K, \forall m \in M$: \\
\vspace*{-10pt}
\begin{enumerate}
\item $msk \leftarrow$ \setup$(1^\lambda,A)$
\item $sk \leftarrow$ \texttt{KeyGen$_{msk}$}$(k)$
\item $\ket{c} \leftarrow$ \texttt{Enc$_{msk}^\mathcal{QE}$}$(m)$
\item $n \leftarrow$ \texttt{Dec$^\mathcal{QD}$}$(sk,\ket{c})$
\end{enumerate}
It is mandated that $msk$ contains $a$, and $n = F_a(k,m)$ with probability 1.
\end{definition}
Note that that hardwiring of the master-secret allows only oracle access to \kgen and \enc. This allows any user to encrypt any message $m$ via an oracle call to get the quantum cipher-text. However in our proofs of security the adversary has only oracle access to \kgen and \enc. Thus the above definition is limited but sound. \\

\subsection{Security Definitions}
\label{subsec:secdef}

We give definitions for quantum IND-secure encryptions, and quantum entropic indistinguishability first. Note that they are reproduced from their original sources.

\begin{definition}[IND-Security, Definition 7, \cite{compsec}]
\label{def:ind-sec}
A secret-key quantum scheme $(\mathcal{QE},\mathcal{QD})$ with secret $s$, has indistinguishable encryptions, or is IND-secure, if for every QPT adversary $(\mathcal{M},\mathcal{D})$, \\
$ \lvert Pr[ \mathcal{D} \{ (\mathcal{QE}_s \otimes \textbf{1}_E) (\rho_{ME}) \} = 1] - Pr[ \mathcal{D} \{ (\mathcal{QE}_s \otimes \textbf{1}_E) (\ket{0}\bra{0}_M \otimes \rho_E) \} = 1] \\ \le $ \negl. \\
where $\rho_{ME} \leftarrow \mathcal{M}(1^\lambda), \rho_E = \textbf{Tr}_M(\rho_{ME})$ and the probabilities are taken over the internal randomness of $\mathcal{QE,M,D}$.
\end{definition}

\begin{definition}[Entropic Indistinguishability, Definition 3, \cite{entsec}]
\label{def:ent-ind}
An encryption scheme with superoperator $\mathcal{E}$ is said to be $(t,\epsilon)$-indistinguishable if for all (density) operators $\rho$ such that $H_\infty(\rho) \geq t$ we have 
$ \lvert\lvert \mathcal{E}(\rho) - \frac{1}{d}\mathbb{I}\rvert\rvert_{tr} \le \epsilon $. Here $d$ is the size of the message space.
\end{definition}

\noindent We next give security definitions for classical functional ciphers adopted from \cite{fesec} in the context of our quantum-classical scheme. Note that superscripts are not exponentiations but indexes. Also 
$\forall b, \texttt{Enc}_{msk,b}(m^0,m^1) = \texttt{Enc}_{msk}(m^b),\\ \texttt{KeyGen}_{msk,b}(f^0,f^1) = \texttt{KeyGen}_{msk}(f^b)$.

\begin{definition}[Valid Message-Privacy Adversary, Definition 2.3, \cite{fesec}]
\label{def:vmpa}
A polynomial-time algorithm $\mathcal{A}$ is a valid message-privacy adversary if for all private-key functional encryption schemes $(\setup,\kgen,\enc,\dec)$ and for all $\lambda \in \mathbb{N}, b \in \{ 0,1 \}$ and for all $f$ and $(m^0,m^1)$ with which $\mathcal{A}$ queries oracles \kgen and $\texttt{Enc}_{msk,b}$ respectively, we have $f(m^0) = f(m^1)$.
\end{definition}

\begin{definition}[Full Message Privacy, Definition 2.4, \cite{fesec}]
\label{def:fmp}
A private-key functional encryption scheme $(\setup,\kgen,\enc,\dec)$ is fully message private if for any valid message-privacy adversary $\mathcal{A}$: \\
$ \lvert Pr[  \mathcal{A}^{\texttt{KeyGen}_{msk}(\cdot),\texttt{Enc}_{msk,0}(\cdot,\cdot)}(\lambda) = 1] - Pr[ \mathcal{A}^{\texttt{KeyGen}_{msk}(\cdot),\texttt{Enc}_{msk,1}(\cdot,\cdot)}(\lambda) = 1] \rvert \\ \le $ \negl.
\end{definition}

\begin{definition}[Valid Function-Privacy Adversary, Definition 3.1, \cite{fesec}]
\label{def:vfpa}
A polynomial-time algorithm $\mathcal{A}$ is a valid function-privacy adversary if for all private-key functional encryption schemes $(\setup,\kgen,\enc,\dec)$ and for all $\lambda \in \mathbb{N}, b \in \{ 0,1 \}$ and for all $(f^0,f^1)$ and $(m^0,m^1)$ with which $\mathcal{A}$ queries oracles $\texttt{KeyGen}_{msk,b}$ and $\texttt{Enc}_{msk,b}$ respectively, we have 
$|m^0| = |m^1|, |f^0| = |f^1|, \text{ and } f^0(m^0) = f^1(m^1)$ where $|\cdot|$ denotes the length of description.
\end{definition}

\begin{definition}[Full Function Privacy, Definition 3.2, \cite{fesec}]
\label{def:ffp}
A private-key functional encryption scheme $(\setup,\kgen,\enc,\dec)$ is fully message private if for any valid function-privacy adversary $\mathcal{A}$: \\
$ \lvert Pr[  \mathcal{A}^{\texttt{KeyGen}_{msk,0}(\cdot),\texttt{Enc}_{msk,0}(\cdot,\cdot)}(\lambda) = 1] - Pr[ \mathcal{A}^{\texttt{KeyGen}_{msk,1}(\cdot),\texttt{Enc}_{msk,1}(\cdot,\cdot)}(\lambda) = 1] \rvert \\ \le $ \negl.
\end{definition}

\noindent Finally we give a scheme security definition from the seminal work \cite{fe}. This has been updated to make the functionality index a part of the comparison (so that we are comparing the same function).

\begin{definition}[Weak Simulation Security, Definition 5, \cite{fe}]
\label{def:wss}
A functional encryption scheme $(\setup,\kgen,\enc,\dec)$ is weakly simulation-secure if for all polynomial-time algorithms (\texttt{Msg},\texttt{Adv}) there exists a polynomial-time algorithm \texttt{Sim} such that the distribution ensembles given in Algorithm \ref{alg:simsec} are computationally indistinguishable.
\end{definition}

\begin{algorithm}[h]
\caption{Weak-Simulation Security Game}
\label{alg:simsec}
  \begin{multicols}{2}
    \begin{algorithmic}[Real]
      \State $msk = a \leftarrow$ \setup$(1^\lambda,A)$
      \State $(\vec m, \tau) \leftarrow \texttt{Msg}(1^\lambda)$
      \State $\ket{\vec c} \leftarrow \texttt{Enc}_{msk}(\vec m)$ via an oracle call
      \State $\alpha \leftarrow \texttt{Adv}^{\texttt{KeyGen}_{msk}(\cdot)}(\ket{\vec c}, \tau)$
      \State Let $(y_1,y_2,...,y_l)$ be the \texttt{Adv} queries
      \State Output real dist. $(a,\vec m, \tau, \alpha, y_1, ..., y_l)$
    \end{algorithmic}
    \columnbreak
    \begin{algorithmic}[Ideal]
     \State Give Functionality-Index $a$ from $A$.
     \State $(\vec m, \tau) \leftarrow \texttt{Msg}(1^\lambda)$
     \State $\alpha \leftarrow \texttt{Sim}^{F_a(\cdot, \vec m)}(1^\lambda, \tau, F_a(\aleph,\vec m))$
     \State Let $(y_1,y_2,...,y_l)$ be the \texttt{Sim} queries to $F_a$
     \State Output ideal dist. $(a,\vec m, \tau, \alpha, y_1, ..., y_l)$
    \end{algorithmic}
  \end{multicols}
\end{algorithm}
\section{A Quantum-Classical Construction}
\label{sec:scheme}

\subsection{The One-Qubit Quantum Cryptosystem}
\label{sec:xi}
We introduce our novel single qubit secret-key quantum cipher $\Xi$ in Algorithm \ref{alg:xi}. Note that the secret key is given by the enclosing hFE cipher. The scheme is proven correct in Theorem \ref{thm:xi-correct} under the appropriate unitary map.

\begin{algorithm}[h]
\caption{The Scheme $\Xi = (\mathcal{QE},\mathcal{QD})$, given the secret-key bit $s$, Bloch-sphere equatorial-position $\theta$, message-bit $b$}
\label{alg:xi}
\begin{algorithmic}[1]

\Procedure{$\mathcal{QE}_{s,\theta}(b)$}{}
\State Prepare $\ket{s}$ from $s$ and $\ket{b}$ from $b$.
\State Sample bit $r \Leftarrow \udist(\{0,1\})$.
\State Return $(\ket{c_0},\ket{c_1}) = (\mathcal{H}^\theta_r\ket{s},\mathcal{H}^\theta_r\ket{b})$.
\EndProcedure

\Procedure{$\mathcal{QD}_{s,\theta}(\ket{c_0},\ket{c_1})$}{}
\State Obtain $\ket{r} = (\mathcal{H}^\theta_s)^\dagger\ket{c_0}$.
\State Prepare $r$ from $\ket{r}$ (by measurement, with probability 1).
\State Obtain $\ket{b} = (\mathcal{H}^\theta_r)^\dagger\ket{c_1}$.
\State Prepare $b$ from $\ket{b}$ (by measurement, with probability 1).
\State Return $b$.
\EndProcedure

\end{algorithmic}
\end{algorithm}

\begin{theorem} [The Quantum Encryption Unitary Map and $\Xi$ Correctness]
\label{thm:xi-correct}
\[ \text{Let } \mathcal{H}^\theta_u := \frac{1}{\sqrt{2}}
\begin{bmatrix}
    1 & 1 \\
    (-1)^u e^{i\theta} & (-1)^{u+1} e^{i\theta} \\
\end{bmatrix}
\text{be the unitary map in scheme $\Xi$.} \]
\noindent Then $\forall \theta,u,v, \mathcal{H}^\theta_u\ket{v} = \mathcal{H}^\theta_v\ket{u}$, and scheme $\Xi$ (in Algorithm \ref{alg:xi}) is correct.
\end{theorem}
\begin{proof}
First, it is an easy verification that 
$\forall \theta,u, (\mathcal{H}^\theta_u)^\dagger\mathcal{H}^\theta_u = \mathcal{H}^\theta_u(\mathcal{H}^\theta_u)^\dagger = \mathbb{I}$. 
Next we have $\mathcal{H}^\theta_u \ket{v} = \frac{1}{\sqrt{2}} \colvec{1}{(-1)^{u \oplus v}e^{i\theta}} = \mathcal{H}^\theta_v \ket{u}$. Finally, to show that the scheme is correct, \\
$\mathcal{QD}_{s,\theta}(\mathcal{QE}_{s,\theta}(b)) = \mathcal{QD}_{s,\theta}(\mathcal{H}^\theta_r\ket{s},\mathcal{H}^\theta_r\ket{b})$
$ = ((\mathcal{H}^\theta_s)^\dagger\mathcal{H}^\theta_r\ket{s},(\mathcal{H}^\theta_?)^\dagger\mathcal{H}^\theta_r\ket{b})$ \\
$ = ((\mathcal{H}^\theta_s)^\dagger\mathcal{H}^\theta_s\ket{r},(\mathcal{H}^\theta_?)^\dagger\mathcal{H}^\theta_r\ket{b})$
$ = (\ket{r},(\mathcal{H}^\theta_?)^\dagger\mathcal{H}^\theta_r\ket{b})$ \\
$ = (\mathcal{H}^\theta_r)^\dagger\mathcal{H}^\theta_r\ket{b}  = \ket{b} = b $ with probability 1, \\
as we obtain the inverting map of $\mathcal{H}^\theta_r\ket{b}$ after recovering $r$.
\end{proof}

\noindent Note that encrypting the randomness under the secret is equivalent to \emph{encrypting the secret under the randomness } - a unique property of the unitary map $\mathcal{H}^\theta_u$.

\subsection{The hFE Scheme}
\label{sec:pi}

Our message (plaintext) length $|m|$ will be greater that or equal to the security parameter. Furthermore, $|m| = Q \in$ \poly. 
We will use $\Sigma = \{ \sigma:\{0,1\}^\lambda \rightarrow \{0,1\}^Q | \text{ } \sigma \text{ is injective} \}$ as part of the functionality index space. Note $|\Sigma| \in $ \expo. Our scheme is given in Algorithm \ref{alg:pi}.

\begin{algorithm}[h]
\caption{The Scheme $\Pi = (\setup,\kgen,\enc,\dec)$, given $k \in \{0,1\}^\lambda, m \in \{0,1\}^Q$, oracle access to $\Xi = (\mathcal{QE},\mathcal{QD})$}
\label{alg:pi}

\begin{algorithmic}

\Procedure{$\setup$}{$1^\lambda,\Sigma$}
\State Sample $a = (\sigma,\kappa_Q) \Leftarrow \udist(\Sigma,\{0,1\}^\lambda)$
\State Return $msk = a$
\EndProcedure

\Procedure{$\kgen_a$}{$k$}
\State Let $\sigma(\kappa_Q) = (s_1,s_2,...,s_Q)$
\State Obtain $\delta \leftarrow \sigma(\kappa_Q) - \sigma(k)$.
\State If $\exists q \in [Q]$ such that $\delta = -\frac{Q-q+1}{2}$ or $\delta = \frac{Q-q}{2}$ \\ return $sk = (s_1,s_2,...,s_q,\bot,\bot,...,\bot)$.
\State Otherwise return $sk = \vec\bot$.
\EndProcedure

\Procedure{$\enc^{\mathcal{QE}}_a$}{$m$}
\State Let $\sigma(\kappa_Q) = (s_1,s_2,...,s_Q)$
\State $\forall j \in [Q], (\ket{c_{j,0}}, \ket{c_{j,1}}) \leftarrow \mathcal{QE}_{s_j,\theta_j}(m_j)$ where $\theta_j = \frac{2 \pi j}{Q}$
\State Return $\ket{\vec c} = (\ket{c_{j,0}}, \ket{c_{j,1}})_{j \in [Q]}$
\EndProcedure

\Procedure{$\dec^{\mathcal{QD}}$}{$sk,\ket{\vec c}$}
\State Let $sk = (s_1,s_2,...,s_q,\bot, \bot, ..., \bot)$ 
\State $\forall j \in [q],  m_j \leftarrow \mathcal{QD}_{s_j,\theta_j}(\ket{c_{j,0}}, \ket{c_{j,1}})$ where $\theta_j = \frac{2 \pi j}{Q}$
\State Return $m_1m_2m_3...m_q$
\EndProcedure

\end{algorithmic}
\end{algorithm}

\noindent We will also use the following nomenclature for our correctness and security proofs.

\begin{definition}
\label{def:key-f}
We say that an arbitrary key $k = \kappa_q$ if there exists (exactly one) $q \in [Q]$ such that $\sigma(\kappa_Q) - \sigma(k) = -\frac{Q-q+1}{2}$ or 
$\sigma(\kappa_Q) - \sigma(k) = \frac{Q-q}{2}$. Also, the function $f_{\kappa_q}$ induced by key $\kappa_q$ is $f_{\kappa_q}(m) = m_1m_2...m_q$.
\end{definition}
\noindent Intuitively, $\kappa_q$ reveals the first $q$-bits of the message, and there is exactly one such $\kappa_q$.

\subsubsection*{Correctness and Positional Secrecy} 
Given for a message $m$, a particular user has key $\kappa_q$, \\
$\dec^{\mathcal{QD}}(\kgen_a(\kappa_q),\enc^{\mathcal{QE}}_a(m)) 
= \dec^{\mathcal{QD}}(s_1s_2...s_q, (\mathcal{QE}_{s_j,\theta_j}(m_j))_{j \in [Q]}) \\
= m_1m_2...m_q$ with probability 1. \\
The different positions of the cipher-qubit on the Bloch sphere equator allow recovery of different subsequences of the message, a notion we call positional secrecy. More formally, let $\eta:\{0,1\}^Q \rightarrow \{0,1\}^Q$ be a permutation that is pre-decided, or optionally is the output of \setup. Then we use $\mathcal{QE}$ to encrypt $m_{\eta(j)}$ under $s_j,\theta_j$. This results in key $\kappa_q$ recovering the $q$-subsequence induced by $\eta$ (rearranging the message bits so that the indices are in order): \\
$\dec^{\mathcal{QD}}(\kgen_a(\kappa_q),\enc^{\mathcal{QE}}_a(m)) 
= \dec^{\mathcal{QD}}(s_1s_2...s_q, (\mathcal{QE}_{s_j,\theta_j}(m_{\eta(j)}))_{j \in [Q]}) \\
= m_{\eta(1)}m_{\eta(2)}...m_{\eta(q)}$ with probability 1. \\

\section{Scheme Security}
\label{sec:security}

We will use the following corollary for our security proofs in this section.
\begin{corollary}
\label{cor:h-zero}
If $\mathcal{H}^\theta_u$ is as defined in Theorem \ref{thm:xi-correct}, then $\forall u,b,\theta \hspace*{5pt} \mathcal{H}^\theta_u\ket{b} = \mathcal{H}^\theta_{u \oplus b}\ket{0}$.
\end{corollary}

\subsection{Security of the Quantum Encryption Scheme}
\label{sec:quan-sec}

We first note that (before prepending with the randomness) $|\mathcal{QE}_{s,\theta}(b)| = 2|b|$ and so the core-function based impossibility result (Definition 6.1,6.2 Theorem 6.3, \cite{semsec}) does not apply to $\Xi$.

\begin{theorem}[Semantic Security]
\label{thm:xi-semsecure}
The one-qubit secret-key encryption scheme $\Xi$ is quantum semantically secure.
\end{theorem}
\begin{proof}
First, we prove that $\Xi = (\mathcal{QE},\mathcal{QD})$ is IND-secure under Definition \ref{def:ind-sec}. Let's say there exists a QPT distinguisher $q\mathcal{D}$ that can distinguish between $\mathcal{QE}_{s,\theta}(\ket{b})$ and $\mathcal{QE}_{s,\theta}(\ket{0})$ in time \poly \text{ } with a probability \nonegl. Consider the following attack where \texttt{cD} is a distinguisher between $r \oplus 0$ and $r \oplus b$, given $s$, $\theta$ and a uniformly random $r$:
\begin{enumerate}
\item The message adversary chooses $b$.
\item Sample $r \Leftarrow \udist(\{0,1\})$.
\item Prepare $\mathcal{H}^\theta_{r \oplus s}\ket{0}$ and send it to $q\mathcal{D}$ as it's first argument.
\item Choose a challenge $c^*$ uniformly between $r \oplus 0$ and $r \oplus b$ and send it to \texttt{cD}.
\item \texttt{cD} prepares $\mathcal{H}^\theta_{c^*}\ket{0}$ and sends it to $q\mathcal{D}$ as it's second argument.
\item $q\mathcal{D}$ now has one of the encryptions between $\mathcal{QE}_{s,\theta}(\ket{0})$ and $\mathcal{QE}_{s,\theta}(\ket{b})$.
\item $q\mathcal{D}$ distinguishes between the two possible encryptions and sends it's decision, $0$ or $b$, to \texttt{cD} (in time \poly \text{ } with a probability \nonegl).
\item \texttt{cD}, for the challenge $c^*$ outputs it's decision ($0$ or $b$), the same as the result given by $q\mathcal{D}$.
\end{enumerate}
Thus we have constructed a computational distinguisher \texttt{cD} between two statistically identical distributions $r \oplus 0$ and $r \oplus b$, even for a $b$ different than $0$, a contradiction. So $q\mathcal{D}$ cannot exist (step 7 cannot happen) and $\Xi$ is IND-secure. \\
We next use Theorem 9 from \cite{compsec} to conclude that $\Xi$ is quantum semantic secure. 
\end{proof}

\begin{theorem}[Entropic Indistinguishability]
\label{thm:xi-entindist}
The one-qubit secret-key encryption scheme $\Xi$ is entropically $(t,\frac{1}{2}(2^{1-t}-1))$-indistinguishable for min-entropy $t \in [0,1]$.
\end{theorem}
\begin{proof}
Here we prove entropic-indistinguishability (Definition \ref{def:ent-ind}) of the message qubit under $\Xi$. Since the secret qubit comes from the uniform distribution, it is perfectly indistinguishable (under the same superoperator) and it's security proof is thus implied. \\
Let $\rho = \sum_{j \in \{0,1\}} \gamma_j \ket{j}\bra{j}$ be the operator corresponding to the message qubit. The only associated interpretations of $\rho$ are classical, according to our scheme. Let $\mathcal{E}$ be the superoperator corresponding to our unitary map. So 
$\mathcal{E}(\rho) := \mathcal{H}^\theta_r \rho (\mathcal{H}^\theta_r)^\dagger$ for uniformly random $r$ and the message space distribution has min-entropy $t = - \log(\max\{\gamma_0,\gamma_1\}) \in [0,1]$. \\
Firstly, it's an easy verification that 
\[ \mathcal{E}(\rho) := \mathcal{H}^\theta_r \rho (\mathcal{H}^\theta_r)^\dagger = \frac{1}{2}
\begin{bmatrix}
    1 & (\gamma_0 - \gamma_1)(-1)^re^{i\theta} \\
    (\gamma_0 - \gamma_1)(-1)^re^{i\theta} & 1 \\
\end{bmatrix} \]
\[ \text{Now, } \lvert\lvert \mathcal{E}(\rho) - \frac{1}{2}\mathbb{I}\rvert\rvert_{tr} = 
\frac{1}{2} \texttt{Tr} \sqrt{(\mathcal{E}(\rho) - \frac{1}{2}\mathbb{I})^\dagger(\mathcal{E}(\rho) - \frac{1}{2}\mathbb{I})} 
= \frac{1}{2} \texttt{Tr}
\begin{bmatrix}
    \frac{\lvert \gamma_0 - \gamma_1\rvert}{2} & 0 \\
    0 & \frac{\lvert \gamma_0 - \gamma_1\rvert}{2}  \\
\end{bmatrix}
 \]
 
 \[ = \frac{\lvert \gamma_0 - \gamma_1\rvert}{2} =  \frac{\lvert 2 \times \max\{\gamma_0,\gamma_1\} - 1 \rvert}{2}  = \frac{\lvert 2 \times 2^{-t} - 1 \rvert}{2} =  \frac{1}{2}(2^{1-t}-1) \]
Thus $\Xi$ is $(t,\frac{1}{2}(2^{1-t}-1))$-indistinguishable.
\end{proof}

\subsection{Security of the hFE scheme}
\label{sec:hFE-sec}

Now we prove security of the classical extension of our quantum encryption using the definitions of privacy/security from the domain of classical functional encryption.

\begin{theorem}[Fully Message Private]
\label{thm:pixi-msgpriv}
The hFE scheme $(\Pi,\Xi)$ is fully message private (under Definition \ref{def:fmp}).
\end{theorem}
\begin{proof}
We start by observing that since the QPT adversary $\mathcal{A}$ is a valid message-privacy adversary (Definition \ref{def:vmpa}), the messages $m^0$ and $m^1$ agree on all bits $[q]$ such that $\mathcal{A}$ queries $\kappa_q$. Let $q^* = max \{q : \mathcal{A} \text{ queries } \kappa_q \}$. Then $\exists j^* \in \{ q^*+1, ..., Q \}$ such that $m^0_{j^*} = 1 - m^1_{j^*}$, otherwise the two encryption oracle calls are identical. \\
Now, let's say $\mathcal{A}$ can distinguish between $\texttt{Enc}_{a,0}(m^0,m^1)$ and $\texttt{Enc}_{a,1}(m^0,m^1)$ in time \poly with a probability \nonegl. Let $\texttt{cD}_j, \forall j \in [Q]$ be  distinguishers on a bit. Let $s \in \{ 0,1 \}^Q$ be the given secret and $\theta_j = \frac{2 \pi j}{Q}, \forall j \in [Q]$ be the given encryption angles. Now consider the following attack for a uniformly random $r \in \{ 0,1 \}^Q$: \\
\vspace*{-8pt}
\begin{enumerate}
\item The adversary chooses messages $m^0,m^1$.
\item Sample $r \Leftarrow \udist(\{0,1\}^Q)$.
\item Prepare $\forall j, \mathcal{H}^{\theta_j}_{r_j \oplus s_j}\ket{0}$ and send them to $\mathcal{A}$ in order as it's odd-position arguments.
\item Choose a challenge bit $b$ uniformly between $0$ and $1$.
\item Send challenge $\forall j, u_j = r_j \oplus m^b_j$ to $\texttt{cD}_j$.
\item Each $\texttt{cD}_j$ prepares $\mathcal{H}^{\theta_j}_{u_j}\ket{0}$ and sends it (in order) to $\mathcal{A}$ as it's even-position arguments.
\item $\mathcal{A}$ now has the encryption $\texttt{Enc}_a(m^b)$.
\item $\mathcal{A}$ outputs $b' = b$ sends it to $\texttt{cD}_j ,\forall j$ (in time \poly \text{ } with a probability \nonegl).
\item $\texttt{cD}_j$, for the challenge $u_j$ outputs it's decision $b'$.
\end{enumerate}
Thus we have constructed a computational distinguisher $\texttt{cD}_{j^*}$ between two statistically identical distributions $r \oplus m^b_{j^*}$ and $r \oplus m^{1-b}_{j^*}$ where 
$m^b_{j^*} = 1 - m^{1-b}_{j^*}$. This is a contradiction. So, $\mathcal{A}$ cannot exist (step 8 cannot happen with a non-negligible probability) and $(\Pi,\Xi)$ is fully message private.
\end{proof}

\begin{corollary}[Fully Function Private]
\label{cor:pixi-funcpriv}
The hFE scheme $(\Pi,\Xi)$ is fully function private (under Definition \ref{def:ffp}).
\end{corollary}
\begin{proof}
Again, on grounds that the QPT adversary $\mathcal{A}$ is a valid function-privacy adversary (Definition \ref{def:vfpa}), 
$f^0_{\kappa_{q^0}}(m^0) = f^1_{\kappa_{q^1}}(m^1) \Rightarrow f^0_{\kappa_{q^0}} = f^1_{\kappa_{q^1}} \Rightarrow \kappa_{q^0} = \kappa_{q^1}$. \\
This means that $\forall q^0,q^1 \in [Q],  \texttt{KeyGen}_{a,0}(\kappa_{q^0},\kappa_{q^1}) = \texttt{KeyGen}_{a,1}(\kappa_{q^0},\kappa_{q^1})$ and the problem reduces to proving full message privacy, which has been proven (Theorem \ref{thm:pixi-msgpriv}). Also, for completeness, every key in $\{ 0,1 \}^\lambda \setminus \{ \kappa_q: q \in [Q] \}$ gives a secret $\vec \bot$, rendering the two $\kgen_{a,u}, u \in \{ 0,1 \}$ oracles same (in output this time). 
\end{proof}

\begin{theorem}[Weakly Simulation Secure]
\label{thm:pixi-weaksim}
The hFE scheme $(\Pi,\Xi)$ is weakly simulation-secure (under Definition \ref{def:wss}).
\end{theorem}
\begin{proof}[Sketch]
Intuitively, we will show that the output distributions of the adversary and simulator are statistically identical (i.e., having zero statistical distance). By virtue of the game, other distributions are same and proving the $\alpha$'s identical is sufficient to prove the real and ideal distribution tuples are statistically identical. \\
First, let us define the following (deterministic) function: \\
$\texttt{S}_a(k) := $ If $k = \kappa_q$ for some $q \in [Q]$, return the first $q$-bits of $\sigma(\kappa_Q)$. Otherwise return $\vec \bot$. \\
And let $\alpha = (\alpha_1,\alpha_2,...,\alpha_l)$ corresponding to key query distributions $( y_j )_{j \in [l]}$. The first thing to note is that the $\alpha_j$'s are distributions on $m_1m_2...m_q$ for $q \in \{ 0 \} \cup [Q]$, and these distributions are a deterministic map, say $\phi$, from $\texttt{S}_a(y_j)$. To see this in the real world game, first observe that since the quantum cipher-text is IND-secure, it does not leak any information about the message to \texttt{Adv}, and so $\alpha^{\textsc{real}}_j = \phi(\texttt{KeyGen}_a(y_j)) =  \phi(\texttt{S}_a(y_j))$. In the ideal world game, firstly, $F_a(\aleph, \vec m) = zQ$ for some integer $z$, leaking nothing additional about the chosen message vector. Also, from the view of comparing the distribution vectors, and since the $\vec m$ distributions are identical by definition, the queries of \texttt{Sim} to $F_a(y_j,\vec m)$ can be reduced to a function of $a$ and $y_j$ alone. That function is exactly $\alpha^{\textsc{ideal}}_j = \phi(\texttt{S}_a(y_j))$ since only one key $\kappa_q$ maps to one $m_1m_2...m_q$, $q \in [Q]$ and all other keys map to the null message. Now let's say there is a non-zero statistical distance $\Delta$ between $\alpha^{\textsc{real}}_j$ and $\alpha^{\textsc{ideal}}_j$. Then 
$0 < \Delta(\alpha^{\textsc{real}}_j,\alpha^{\textsc{ideal}}_j) = \Delta(\texttt{S}^{\textsc{real}}_a(y_j),\texttt{S}^{\textsc{ideal}}_a(y_j)) $ \\
$ \le \Delta(y_j, y_j)$ (as $\texttt{S}^{\textsc{real}}_a = \texttt{S}^{\textsc{ideal}}_a = \texttt{S}_a$, applying Theorem 7.6, \cite{statdist}) \\
$\Rightarrow 0 < \Delta(y_j, y_j)$. \\
Thus we have a contradiction. So $\forall j \in [l], \alpha^{\textsc{real}}_j = \alpha^{\textsc{ideal}}_j$, and the real and ideal distribution ensembles are statistically indistinguishable.
\end{proof}

\section{Discussion}

We see that the scheme $(\Pi,\Xi)$ permits an arbitrary, polynomial-sized stretch in the length of the message, given the security parameter. This is realizable due to the infinitely many cipher positions permissible on the Bloch-sphere equator. Note that this is not possible classically. \\
Also, there needs to be consideration on efficient representations of arbitrary injective functions from $\lambda$-bits to $Q$-bits, which in general have exponential size tables. Note that for the purpose of our scheme (where $[[x]]$ denotes an efficient representation of $x$), \\ 
$[[a]] := (\sigma(\kappa_Q), \kappa_1, \kappa_2, ..., \kappa_Q)$ should suffice (see Definition \ref{def:key-f}). \\
Finally, it's important to observe that the keys $k$ should be distributed only after instantiating $a$, that is, running \texttt{Setup}. This is true because under different $\sigma$'s, the same key $k$ does not necessarily give the same decryption.
\section{Related Work}

Over the past twenty years, there has been a lot of work done on quantum ciphers. Most of them encrypt pure quantum states as opposed to classical information. Beginning with private quantum channels \cite{pqc} whose optimality was proved later \cite{pqc-optimality}, quantum vernam ciphers were developed \cite{qe-vernam}. An optimal scheme based on quantum one-time pads \cite{qe-optimal} was proposed, later followed by characterisation of a one-way quantum encryption scheme \cite{qe-perfect} more general than private quantum channels. \\
Initially, Zhou proposed an algorithm to encrypt binary classical information \cite{qe-realqu}. Zhou then proposed qubit block encryption algorithms \cite{qbe-zhou,qbe-hk-zhou} which were later improved \cite{qe-improve-cao} to base their security on the BB84 protocol. Other than these, there have been parallel works on symmetric-key schemes \cite{qe-symmetric,qe-symm2}, a $d$-level systems' scheme \cite{qe-dlev}, and schemes based on modified BB84 \cite{qe-modbb84}, quantum key generation \cite{qbe-qkg}, conjugate coding \cite{qe-prob}, and quantum shift registers / hill cipher \cite{qe-qsrhc}. \\
Quantum public-key encryption (QPKE) schemes have been addressed by the community since the beginning of this century \cite{qpke-okamoto}. There have been QPKE schemes with information theoretic security proposed by Pan \cite{qpke-itsec-pan} and Liang \cite{qpke-itsec-liang,qpke-itsec2-liang}. Also, there have been QPKE schemes based on single qubit rotations \cite{qpke-qurot,qpke-qurot-sec} and classical NP-complete problems \cite{qpke-npc}. \\
More recently, quantum (fully) homomorphic encryption (Q(F)HE) schemes have been developed by the community. Liang proposed a perfectly secure QFHE scheme based on the quantum one-time pad \cite{qfhe-perfect}. This was followed by a QFHE scheme based on the universal quantum circuit \cite{qfhe-uqc}. Also, there has been a QHE scheme for polynomial sized circuits given by Dulek \cite{qhe-polyc}. \\
Very early, a KCQ (keyed communication in quantum noise) approach to cryptography \cite{kcq} was presented. Other works in the space of conventional quantum encryption include studies on optimality of quantum encryption schemes \cite{qe-optimality}, the use of quantum keys as opposed to classical keys \cite{qe-quantkeys}, and non-malleable ciphers \cite{nonmalleable}. Quantum secure direct communication \cite{qsdc-qe} and quantum key distribution \cite{qkd-via-qe} via quantum encryption have also been proposed. Finally, other notions of security \cite{qe-semsec-xiang,oram-sec} have been professed.
\section{Conclusions and Future Work}

In this work, we have introduced a novel one-qubit secret-key quantum encryption scheme for classical information. We have proved this scheme to have quantum semantic security, and quantum entropic indistinguishability as a function of the min-entropy of the message distribution. We have extended this scheme to permit recovery of different length subsequences of the message using different keys, under a new notion of positional secrecy. The resulting (hybrid) functional encryption scheme is proven to be full-message private, full-function private and weakly simulation-secure. \\
We hope to see the following improvements in the future, given the current status of our quantum-classical scheme: \\
\begin{enumerate}
\vspace*{-10pt}
\item Given that, under $\Xi$, encryptions of both $0$ and $1$ are statistically indistinguishable, perhaps there exists a proof of entropic security (Definition 2, \cite{entsec}) for the quantum scheme.
\item We recognize that the biggest drawback of the classical extension is that \kgen and \enc algorithms are only available via oracle calls, although that does not affect the security proofs. There could be a modification which permits making these algorithms public. 
\item We hope that this work motivates a new general definition for quantum functional encryption - fully quantum schemes which permit learning meaningful functions from encryptions of (general) quantum states.
\end{enumerate}

\subsubsection*{Acknowledgements}
We would like to thank Dr. Vinay J. Ribeiro for constructive discussions.


\end{document}